\documentclass[12pt,a4paper]{article}
\usepackage{amssymb}
\usepackage{amsmath,amssymb}
\usepackage[all]{xy}
 \usepackage{latexsym}
\usepackage{amsthm,a4wide}
\input diagxy

  \theoremstyle{plain}

  \newtheorem{thm}{Theorem}[section]
  \newtheorem{lem}[thm]{Lemma}
  \newtheorem{example}[thm]{Example}

  \theoremstyle{definition}

\begin{document}
\newcommand{\bm}{\bibitem}

\title{{\ Infinitely Many Periodic Solutions for Some  N-Body\\ Type Problems with Fixed Energies\thanks{Supported
partially by NSF of China.}}}
\author{ {\normalsize   Pengfei Yuan and Shiqing Zhang}\\
{\small Department of Mathematics, Sichuan University, Chengdu
610064, China}}
\date{}
\maketitle
\begin{abstract}In this paper, we apply the Ljusternik-Schnirelman theory with local
Palais-Smale condition to study a class of N-body problems with
strong force potentials and fixed energies. Under suitable
conditions on  the potential $V$, we
prove  the  existence of infinitely many non-constant and non-collision symmetrical periodic solutions .  \\[5pt]

\\[4pt]
\bf{2000 Mathematicals Subject Classification: 34C15, 34C25, 58F}
\end{abstract}
%\tableofcontents
\parskip 1pt

\section{Introduction and  Main Results}%\baselineskip 15pt
\hspace*{\parindent} The N-body problem  is an old and very
difficult problem. Many  mathematicians studied this problem using
many different methods, here we only concern variational methods. In
1975 and 1977, Gordon $([15],[16])$ firstly used variational methods
to study the periodic solutions of 2-body problems. In $[16]$, he
put forward to the strong force condition (SF for short) and got the
following Lemma:
\begin{lem} (Gordon$[16]$) \quad Suppose that $V(x)$ satisfies the so called
Gordon's Strong Force condition: There exists a neighborhood
$\mathcal{N}$ of 0 and a
 function $U\in C^1 (\mathcal{N}\setminus\{0\},\mathbb{R})$
  such that:\\
$(i).\: \underset{|x| \rightarrow 0 }{\lim}U(x)  =-\infty ;$ \\
$(ii).\: -V(x)\geq |\nabla U(x)|^2$ for every $x\in
\mathcal{N}\,\backslash\,\{0\}. $

Let
\begin{align*}
\Lambda &\triangleq \{x |\, x(t)\in
H^1(\mathbb{R}/\mathbb{Z},\mathbb{R}^n),\, x(t)\neq 0, \forall\,
t\,\in [0,1]\}  \\
\partial\Lambda&\triangleq \{x |\, x(t)\in H^1(\mathbb{R}/\mathbb{Z},\mathbb{R}^n),\, \exists\, t_0 \in [0,1]\,
s.t.\,x(t_0)= 0
 \}
\end{align*}
 Then we have
$$
\int_0^1V(x_n)\,\mathrm{d}t \rightarrow -\infty,\quad \forall \: x_n
\rightharpoonup x \in \partial \Lambda.
$$
\end{lem}
After Gordon, many researchers used variational methods to study N
($N\geq 3$) body problems $([1]-[8], [10]-[19],[22]-[27],etc.)$.

In $[18]$, P.Majer used Ljusternik-Schnirelman  theory (LS for
short) with local Palais-Smale condition  to seek T-periodic
solutions of the following second order Hamiltonian system:
$$
\ddot{x}+ax+\nabla_x W(t,x)=0 \eqno{(1.1)}
$$
where $W$ is singular at $x=0, W(t+T,x)=W(t,x).$

He got the following  theorem :
\begin{thm}
Suppose

$\quad(P_1). a<(\dfrac{\pi}{T})^2;$

$\quad(P_2). W\in C^1(S_T^1\times (\mathbb{R}^n\setminus
\{0\}),\mathbb{R})$ satisfies $(SF)$;

$\quad(P_3). \exists \, c ,\, \theta<2,\, r>0 $, such that
$\forall\,|x|\geq r, \forall\, t\in S_T^1$
$$
W(t,x)\leq c|x|^\theta, \quad \langle \nabla W_x(t,x),x
\rangle-2W(t,x)\leq c|x|^\theta
$$

$\quad(P_4). W(t,x)\leq b .$

Then the system $(1.1)$ has infinitely many T-periodic non-collision
solutions .
\end{thm}

After Majer , Zhang-Zhou $([27])$ studied  a class of N-Body
problems. They considered the following system :
$$
m_i\ddot{x}(t)+\nabla_{x_i}V(t,x_1(t),\cdots,x_N(t))=0,\, x_i\in
\mathbb{R}^n,\,i=1,\cdots,N \eqno{(1.2)}
$$
where $m_i>0$ for all $i$, and $V$ satisfies the following
conditions :

$\quad(V1). V(t, x_1,\cdots,x_N)=\dfrac{1}{2}\underset{1\leq i\neq j
\leq N}{\sum}V_{ij}(t,x_i-x_j);$

$\quad(V2). V_{ij}\in C^2(\mathbb{R}\times
(\mathbb{R}^n\setminus\{0\});\mathbb{R}),$ for all $1\leq i\neq
j\leq N $;

$\quad(V3). V_{ij}(t,\xi) \rightarrow -\infty $ uniformly on $t$ as
$|\xi|\rightarrow 0$,\,\,for all $1\leq i\neq j\leq N $;

$\quad(V4). V_{ij}(t,\xi)\leq 0$,\quad for all $t\in \mathbb{R},\,
\xi\in \mathbb{R}^n\setminus\{0\}$;

$\quad(V5). V_{ij}$ satisfies Gordon's strong force condition;

$\quad(V6)$. There exists an element $g$ of finite order $s$ in
$SO(k)$ which has no fixed point other than origin (i.e. 1 is not an
eigenvalue of g), such that
$$
V(t, x_1,\cdots, x_N)=V(t+T/s,
gx_1,\cdots,gx_N),\,\forall\,t\in[0,1] ,\,x_i\in \mathbb{R}^n.
$$
If the potential $V$ satisfies $(V1)-(V6)$, and
$x(t)=(x_1(t),\cdots,x_N(t))$ is a T-periodic non-collision solution
of system  $(1.2)$ and satisfies
$x(t+T/s)=(x_1(t+T/s),\cdots,x_N(T+t/s))=(gx_1(t),\cdots,gx_N(t))=gx(t)$,
we
say that $x(t)$ is a g-symmetric T-periodic non-collision solutions
(we also call $(g,T)$ is a non-collision solution for short).

 Zhang-Zhou obtained the following theorem:
\begin{thm}
Suppose the potential $V$ satisfies $(V1)-(V6)$ and $T$ is any
positive  real number. Then system $(1.2)$ has infinitely many
$(g,T)$ non-collision solutions.
\end{thm}

Motivated by P.Majer, Zhang S.Q.- Zhou Q. 's work , we consider the
following system:
$$
\left\{
\begin{array}{cc}
 &m_i\ddot{u}_i + \nabla_{u_i} V(u_1,\ldots,u_N) =0,\quad(1\leq i\leq N),\quad (Ph.1) \\
 &\dfrac{1}{2}\sum m_i|\dot{u}_i(t)|^2
 +V\left(u_1(t),\ldots,u_N(t)\right)=h.\quad \qquad (Ph.2)
\end{array}
\right.\leqno{(Ph)}
$$
We have the following theorem :
\begin{thm}
Suppose  $V( u_1,\cdots,u_N)=\dfrac{1}{2}\underset{1\leq i\neq j
\leq N}{\sum}V_{ij}(u_i-u_j),\,\, V_{ij}(\xi)\in
C^1(\mathbb{R}^n\setminus\{0\},\mathbb{R})$, and satisfies

$\quad(V5').V_{ij}(\xi)\leq 0,\forall\, \xi\in
\mathbb{R}^n\setminus\{0\}$;

$\quad(V6').\exists\, \alpha>2$ and $r_1>0$ such that $\langle\nabla
V_{ij}(\xi),\xi\rangle \geq -\alpha V_{ij}(\xi),\forall 0<|\xi|\leq
r_1;$

$\quad(V7').\exists\, c \geq 0, -2 <\theta<0, r_2>r_1$, such that
when $|\xi|\geq r_2 $, there holds
$$
\langle\nabla_\xi V_{ij}(\xi), \xi \rangle\leq c |\xi|^\theta;
$$

$\quad(V8').V_{ij}(\xi)=V_{ji}(\xi),\forall \,\xi\in
\mathbb{R}^n\setminus\{0\}. $

Then  for any $h> 0$, the system $(Ph)$ has infinitely many
non-constant  and  non-collision periodic solutions .
\end{thm}
\begin{example}
We take
$$
V_{ij}(\xi)=\left\{
\begin{array}{cc}
-\dfrac{1}{|\xi|^\alpha},\ \ & \ \ {\rm if} \ \ 0<|\xi|\leq r_1 ,\\
smooth \,\,connecting ,\ \ & \ \ {\rm if} \ \ r_1<|\xi|<r_2 , \\
-|\xi|^\theta ,\ \ & \ \ {\rm if} \ \ |\xi|\geq r_2>r_1.
\end{array}
\right.
$$
\end{example}

\section{ Some Lemmas}
In this section, we collect some known  lemmas which are necessary
for  the proof of Theorem 1.4.

Let us introduce the following notations:
\begin{align}
&m^*=\min{\{m_1,\cdots,m_N\}}\:;\qquad  H^1=W^{1,2}(\mathbb{R}/\mathbb{Z},\mathbb{R}^n).     \notag    \\
&E=\{u=(u_1,\ldots,u_N)\mid \ u_i\in H^1,\, u_i(t+\dfrac{1}{2})=-u_i(t)\}.    \notag \\
&\Lambda_0=\{u\in E\mid u_i(t)\neq
u_j(t),\:\forall\,t,\,\forall\,i\neq j \}.   \notag \\
&\partial\Lambda_0=\{u\in E \mid \exists\, t_0,1\leq i_0\neq j_0\leq
N\,s.t.\,u_{i_0}(t_0)=u_{j_0}(t_0)\}. \notag   \\
&p(u)=\underset{1\leq i\neq j\leq N,\,t\in [0,1]
}{\min{}}|u_i(t)-u_j((t)|   \notag.\\
&\{f\leq c\}=\{ u\in \Lambda_0,\,f(u)\leq c \}.    \notag
\end{align}

\begin{lem} \label{2.1}
$([1]-[4])$\quad Let
$f(u)=\dfrac{1}{2}\int_0^1\sum\limits_{i=1}^{N}m_i|u_i|^2
\,\mathrm{d}t\int_0^1(h-V(u))\, \mathrm{d}t$ and $\tilde{u}\in
\Lambda_{0}$ satisfy $f'\left(\tilde{u}\right)=0$ and
$f\left(\tilde{u}\right)>0$. Set
$$
\dfrac{1}{T^2}=\dfrac{\int_0^1\left(h-V\left(\tilde{u}\right)\right)\,
\mathrm{d}t}
{\dfrac{1}{2}\int_0^1\sum\limits_{i=1}^Nm_i|\dot{\tilde{u}}_i|^2
\,\mathrm{d}t} .                                    \eqno{(2.1)}
$$
Then $\tilde{q}(t)=\tilde{u}\left(t/T\right)$ is a non-constant
$T$-periodic solution for (Ph).
\end{lem}

\begin{lem} \label{2.2} $(Palais[20])$ \quad Let $\sigma$ be an orthogonal
representation of a finite or compact group $G$ in the real Hilbert
space $H$ such that for $\forall\:\sigma\in \, G$, $$f(\sigma \cdot
x)=f(x),$$ where $f\in C^1(H,\mathbb{R})$.

Let $S=\{x\in H \mid \sigma x=x,\:\forall\: \sigma\ in \ G \}$, then
the critical point of $f$ in $S$ is also a critical point of $f$ in
$H$.

\end{lem}
By Lemma 2.1-2.2, we have
\begin{lem} \label{2.3}$([1]-[4])$\quad Assume $V_{ij}\in C^1(\mathbb{R}^n\setminus\{0\},\mathbb{R})$
satisfies $(V8').$ If $ \bar{u} \in \Lambda_0$ is a critical point
of $f(u)$ and $f(\bar{u})>0$, then $\bar{q}(t)=\bar{u}(t/T)$ is a
non-constant $T$-periodic
solution of (Ph). \\[4pt]
\end{lem}
%\begin{lem} \label{2.7}(Sobolev-Rellich-Kondrachov, Compact Imbedding Theorem
%$[28]$)
%$$
%W^{1,2}(\mathbb{R} / T\mathbb{Z},\mathbb{R}^n) \subset C(\mathbb{R}
%/ T\mathbb{Z},\mathbb{R}^n)
%$$
% and the embedding is compact.
%\end{lem}
%
%\begin{lem} \label{2.8}
%(Eberlein-Shmulyan$[4]$)\quad A Banach space $X$ is reflexive if and
%only if any bounded sequence in $X$ has a weakly convergent
%subsequence.
%\end{lem}
\begin{lem} \label{2.9} ($[28]$)\quad Let $q \in W^{1,2}(\mathbb{R} / T\mathbb{Z},\mathbb{R}^n) $
and $ \int_0^Tq(t)\,\mathrm{d}t=0 $, then we have \\
$(i)$.\,Poincare-Wirtinger's inequality:
$$
\int_0^T|\dot{q}(t)|^2\,\mathrm{d}t \geq
{\Big(\frac{2\pi}{T}\Big)}^2 \int_0^T|{q}(t)|^2\,\mathrm{d}t.
\eqno{(2.2)}
$$
$(ii)$.\,  Sobolev's inequality:
$$
\underset{0 \leq t \leq T} {\max} |q(t)|
 = |q|_\infty  \leq  \sqrt {\frac{T}{12}}\big(\int_0^T |\dot
 {q}(t)|^2\mathrm{d}t\big)^{1/2}. \eqno{(2.3)}
$$
\end{lem}
By the definition of $\Lambda_0$ and Lemma 2.3,  for $ \forall\, u
\in \Lambda_0, (\int_0^1\sum\limits_{i=1}^N m_i|\dot{u}_i|^2
\,\mathrm{d}t)^{1/2}$ is equivalent to the $(H^1)^N=H^1\times \cdots
H^1$ norm:
$$
\parallel u \parallel_{(H^1)^N}=\left(\int_0^1|\dot{u}|^2 \,\mathrm{d}t\right)^{1/2}
+\left(\int_0^1 | u|^2\,\mathrm{d}t\right)^{1/2}.
$$
So we take norm $\parallel u \parallel =(\int_0^1\sum\limits_{i=1}^N
m_i|\dot{u}_i|^2 \,\mathrm{d}t)^{1/2}$.

\begin{lem}
(Coti Zelati$\,[11]$)\quad Let $X=(x_1,\ldots,x_N)\in
\mathbb{R}^n\times\cdots \mathbb{R}^n, x_i\neq x_j,\,\forall\,1\leq
i\neq j\leq N .$ Then
$$
\sum\limits_{1\leq i<j\leq N}\dfrac{m_im_i}{|x_i-x_j|^\alpha}\geq
C_\alpha(m_1,\ldots,m_N) (\sum\limits_{i=1}^N
m_i|x_i|^2)^{-\alpha/2},
$$
where  $
C_\alpha(m_1,\ldots,m_N)\overset{\triangle}{=}C_\alpha=(\sum\limits_{i=1}^N
m_i)^{-\alpha/2}(\sum\limits_{1\leq i<j\leq
N}m_im_j)^{\frac{2+\alpha}{2}}.$
\end{lem}

\begin{lem}$([18])$ \quad Let $X$ be a Banach space with norm $\parallel \cdot
\parallel$, $\Lambda$ be an open subset of $X$, and suppose a
functional $f: \Lambda \rightarrow \mathbb{R} $ is given such that
the following conditions hold:

$\quad(i). Cat_{\Lambda}(\Lambda)=+\infty;$

$\quad(ii). f\in C^1(\Lambda)$ and $\forall u_n\rightharpoonup
\partial\Lambda, f(u_n)\rightarrow +\infty;$

$\quad(iii). \forall \lambda \in \mathbb{R}, Cat_{\Lambda}(\{f\leq
\lambda\})< +\infty;$

Suppose in addition that there exist $g\in C^1(\Lambda), \beta\in
(0,1)$ and $\lambda_0\in \mathbb{R}$ such that

$\quad(iv). Cat_{\Lambda}(\{f\leq g\}) < +\infty$;

$\quad(v).$ the PS condition holds in the set $\{f\geq g\};$

$\quad(vi). \beta\parallel f'(u) \parallel \geq \parallel
 g'(u)\parallel, \forall u\in \{f=g \geq \lambda_0\}.$

 Then $f$ has a sequence $\{u_n\}\subset \Lambda$  of critical
 points such that $f(u_n)\rightarrow +\infty $  and $f(u_n)\geq g(u_n)-1.$
\end{lem}

%It's not difficult to prove:
%\begin{lem} \label{2.10}  For $ \forall\, u \in \Lambda_0 $, we have
%$$
%\int_0^1u(t) \,\mathrm{d}t=0.
%$$
%By Lemma 2.5 and Lemma 2.6, for $ \forall\, u \in \Lambda_0,
%\Arrowvert u\Arrowvert =(\int_0^1\sum\limits_{i=1}^N
%m_i|\dot{u}_i|^2 \,\mathrm{d}t)^{1/2}$ is equivalent to the $H^N$
%norm:
%$$
%\parallel u \parallel_{H^N}=\left(\int_0^1|\dot{u}|^2 \,\mathrm{d}t\right)^{1/2}
%+\left(|\int_0^1 u\,\mathrm{d}t|\right)^{1/2}.
%$$
%\end{lem}

\section{The Proof of Theorem 1.4}

\begin{lem}
\quad Let $\{u_n\} \subset \Lambda_0$ and $u_n\rightharpoonup u\in
\partial\Lambda_0. $ Then $f(u_n)\rightarrow +\infty.$
\end{lem}
\begin{proof}
First of all , we recall that
$$
f(u_n)=\dfrac{1}{2}\int_0^1\sum_{i=1}^{N}m_{i}|\dot{u}_n^i|^2dt\int_0^1(h-V(u_n))dt
$$
(1).If $u= $ constant, we can deduce $u=0$ by
$u_i(t+\dfrac{1}{2})=-u_i(t)$. By Sobolev's embedding theorem, we
obtain
$$|u_n| _\infty\rightarrow 0 ,\quad
n\rightarrow   \infty \eqno{(3.1)}
$$
So when $n$ is large enough, $0\leq |u_n^i-u_n^j|\leq r_1$.  By
$(V6')$, there exists an $A>0$ such that
$$
V_{ij}(\xi)\leq -\dfrac{A}{|\xi|^\alpha}, \quad \forall 0<|\xi|\leq
r_1 \eqno{(3.2)}
$$

By $h>0$, Sobolev's inequality $(2.3)$ , $(3.2)$ and Lemma 2.5,  we
have
\begin{align*}
f(u_n)&=\dfrac{1}{2}\parallel u_n \parallel^2\int_0^1(h-\sum_{1\leq
i<j\leq N}V_{ij}(u_n^i-u_n^j))dt  \\
&\geq \dfrac{1}{2}\parallel u_n \parallel^2\int_0^1\{-\sum_{1\leq
i< j\leq N}V_{ij}(u_n^i-u_n^j)\}dt  \\
&\geq \dfrac{1}{2}\parallel u_n \parallel^2 \int_0^1\sum_{1\leq
i< j\leq N}\dfrac{A}{|u_n^i-u_n^j|^\alpha}dt \\
 &\geq \dfrac{1}{2}\parallel u_n \parallel^2\int_0^1
A[\dfrac{N(N-1)}{2}]^{\frac{2+\alpha}{2}}N^{\frac{-\alpha}{2}}|u_n|^{-\alpha}dt  \\
&\geq 3m^*A2^{-\frac{\alpha}{2}}N(N-1)^{1+\frac{\alpha}{2}} |u_n|
_\infty ^{2-\alpha}
\end{align*}
Then by $(3.1)$ and $\alpha >2$, we can deduce
$$
f(u_n)\rightarrow +\infty,\quad n\rightarrow \infty
$$
(2).\:If $u\not\equiv$ constant, by the weakly lower-semi-continuity
property for norm, we have
\begin{align}
\underset{n\rightarrow\infty}{\liminf}\int_0^1\sum\limits_{i=1}^N
m_i|\dot{u}_n^i|^2 \,\mathrm{d}t & \geq  \int_0^1\sum\limits_{i=1}^N
m_i|\dot{u}_i|^2\,\mathrm{d}t >0. \tag{3.3}
\end{align}
Since $u\in\partial\Lambda_0$, there exist $t_0,1\leq i_0\neq
j_0\leq N$\,s.t.\,$u_{i_0}(t_0)=u_{j_0}(t_0).$ Set
\begin{align}
\xi_n(t)=u_n^{i_0}(t)-u_n^{j_0}(t)   \notag \\
\xi(t)=u_{i_0}(t)-u_{j_0}(t)   \notag
\end{align}
By $u_n\rightharpoonup u$, we have $\xi_n(t)\rightharpoonup\xi(t)$.
Then by $(V6')$ and Lemma 1.1, we have
$$
\int_0^1V_{
i_0j_0}(u_n^{i_0}-u_n^{j_0})\,\mathrm{d}t\rightarrow-\infty.
$$
Recalling that
$$
V(u_n)=\sum\limits_{i<j}V_{ ij}(u_n^i-u_n^j).
$$
So we have
$$
f(u_n)\rightarrow+\infty,\quad n\rightarrow\infty.
$$
\end{proof}

\begin{lem}
For every $\lambda\in \mathbb{R}$ there exists a constant
$k=k(\lambda)$ such that
$$
\parallel u \parallel  \leq k(\lambda)p(u),\quad \forall u\in \{f\leq
\lambda\}   \eqno{(3.4)}
$$
\end{lem}

\begin{proof}
Recall that
$$
f(u)=\dfrac{1}{2} \parallel u \parallel^2\int_0^1(h-V(u))dt
$$
Since $V_{ij}\leq 0$, we have $f(u)\geq\dfrac{1}{2}h\parallel u
\parallel^2$. Since $u\in \{f\leq \lambda\}$, we can deduce that
$$
\parallel u \parallel\leq C.  \eqno{(3.5)}
$$
If $(3.4)$ is not true, then there exists a sequence $\{u_n\}$ such
that $\parallel u_n \parallel \geq n p(u_n)$. By $(3.5)$, we have
$$
0<p(u_n)\leq \dfrac{C}{n} .   \eqno{(3.6)}
$$
Let $n \rightarrow \infty $ in (3.6), we have $p(u_n)\rightarrow 0.$
Then there exists a subsequence $\{u_n\}\rightharpoonup  u\in
\partial\Lambda_0$, by Lemma 3.1, $f(u_n)\rightarrow +\infty $, which
is a contradiction since $\{u_n\}\subset\{f\leq \lambda\}.$

\end{proof}
\begin{lem}
\quad For every $\lambda\in \mathbb{R}$, the set
$\Lambda_\lambda=\{u\in \Lambda_0 :\dfrac{\parallel u \parallel
}{p(u)}\leq \lambda \}$ is of finite category in $\Lambda_0$.
\end{lem}
\begin{proof}
The proof is almost the same as the proof of  Lemma 4.3 of [27]. For
the covenience of the readers, we write the complete proof.

 It is
sufficient that we give a homotopy $h: [0,1]\times \Lambda_c\subset
\subset \Lambda_0 .$

Take $\delta\in(0,1)$ s.t. $2\dfrac{1}{m^*}\sqrt{\delta}c<1$, and
define $\phi(t)=\dfrac{p(u}{\delta}$, if $t\in[0,\delta]$,otherwise
$\phi(t)=0.$

Define
$$
h(u,\lambda)=(1-\lambda)u(t)+\lambda\dfrac{(u\ast \phi)(t)}{p(u)},
\quad \forall u\in \Lambda_c, \,0\leq\lambda\leq 1
$$
where the convolution  $(u\ast
\phi)(t)=(\int_0^1u_1\phi(s)ds,\cdots,\int_0^1u_N\phi(s)ds)$, then
$h(u,0)$ is an inclusion and $h_1(\Lambda_c)$ is paracompact since
$h(u,1)$ is a convolution operator , so it's a compact operator. We
need to prove $h(\Lambda_c\times I)\subset\Lambda_0.$

Suppose this is not true , then
$\exists\lambda_0\in(0,1],\,u\in\Lambda_c,\,1\leq i_0\neq j_0\leq
N,\,t_0\in [0,1]$, such that
$$
(1-\lambda_0)(u_{i_0}-u_{j_0})+\lambda_0\dfrac{1}{p(u)}((u_{i_0}-u_{j_0})\ast\phi
)(t_0)
$$
Then
$$
(u_{i_0}-u_{j_0)}\ast\phi(t_0)=p(u)(1-\dfrac{1}{\lambda_0})(u_{i_0}(t_0)-u_{j_0}(t_0))
$$
So we have
\begin{align*}
&|p(u)(u_{i_0}(t_0)-u_{j_0}(t_0)-((u_{i_0}-u_{j_0}\ast\phi))(t_0))|
\\
&=\dfrac{p(u)}{\lambda_0} \geq p(u)p(u)=p^2(u).
\end{align*}
On the other hand , $\forall u=(u_1,\cdots,u_N)\in \Lambda_c
,\,i\neq j,\,t\in[0,1]$, we have
\begin{align*}
&|p(u)(u_i(t)-u_j(t))-((u_i-u_j)\ast\phi)(t)|  \\
&=|\dfrac{p(u)}{\delta}\int_0^\delta(u_i(t)-u_j(t))ds-
\dfrac{p(u)}{\delta}(u_i(t-s)-u_j(t-s))ds| \\
&\leq
\dfrac{p(u)}{\delta}\int_0^\delta|u_i(t)-u_j(t)-(u_i(t-s)-u_j(t-s))|ds
\\
&\leq p(u)\cdot \underset{|s|\leq
\delta}{\sup{}}|u_i(t)-u_j(t)-(u_i(t-s)-u_j(t-s))|\\
& \leq p(u)\sqrt{\delta}\parallel
\dot{u}_i-\dot{u}_j\parallel_{L^2}\leq
p(u)\sqrt{\delta}(\parallel\dot{u}_i
\parallel_{L^2}+\parallel\dot{u}_j
\parallel_{L^2} ) \\
&\leq p(u)2\sqrt{\delta}\parallel\dot{u}
\parallel_{L^2}\leq p(u)\cdot (2\dfrac{1}{m^*}\sqrt{\delta}c)\cdot p(u)<p^2(u)
\end{align*}
Which is a contradiction.
\end{proof}

\begin{lem}
The functional $f$ verifies the Palais-Smale condition on
$\Lambda_0$.
\end{lem}
\begin{proof}
Let $\{u_n\}\subset\Lambda_0 $ be a P.S. sequence, then up to a
subsequence, it converges weakly in $(H^1)^N $ and uniformly in
$|u|_\infty$ to an element $u \in\Lambda_0$ by Lemma 3.1. Hence
$\langle\nabla_{\xi}V_{ij}(u_i^n-u_j^n),(u_i-u_i^n-u_j-u_j^n)\rangle$
converges uniformly to zero. Since $f^\prime(u_n)\rightarrow 0$, and
$u-u_n$ is $(H^1)^N$-bounded, and
$$
\langle f'(u), v\rangle=\int_0^1\sum_{1\leq i\leq N}m_i\langle
\dot{u_i},\dot{v_i}\rangle dt
\int_0^1(h-V(u))dt-\dfrac{1}{2}\parallel u
\parallel^2\int_0^1\langle \nabla_uV(u),v\rangle dt
$$
So we have
\begin{align*}
\parallel u \parallel^2-\underset{n\rightarrow \infty}{\lim{}}\parallel u_n
\parallel^2&=\underset{n\rightarrow \infty}{\lim{}}\int_0^1
m_i\langle  \dot{u}_n, (\dot{u}-\dot{u_n})\rangle dt\\
&=\underset{n\rightarrow \infty}\lim{}\dfrac{\langle f'(u_n) , u-u_n
\rangle }{\int_0^1(h-V(u_n))dt} +\underset{n\rightarrow \infty
}\lim{}\dfrac{\frac{1}{2}\parallel u_n
\parallel^2 \int_0^1 \langle \nabla_{u}V(u_n),(u-u_n)\rangle
dt}{\int_0^1(h-V(u_n))dt}=0
\end{align*}
\end{proof}
\begin{lem}
$Cat_{\Lambda_0}(\Lambda_0)=+\infty$
\end{lem}
\begin{proof}
See the Corollary 3.4 of [27].
\end{proof}
\begin{lem}
There exist  a  functional $g\in C^1(\Lambda_0),\beta\in (0,1)$ and
$\lambda_0 \in \mathbb{R}$ such that:

(iv)$Cat_{\Lambda_0}\{f\leq g\} <
 + \infty;$

(v) the P.S. condition holds in the set $\{f\leq g\};$

(vi)$\beta \parallel f'(u)\parallel \geq \parallel
g'(u)\parallel,\quad \forall u \in \{f=g\geq \lambda_0 \}.$

\end{lem}

\begin{proof}
By Sobolev's embedding theorem, we know there is a $k_\infty>0$ s.t.
$$
\parallel u \parallel_\infty\leq k_\infty \parallel u \parallel,
\forall\, u\in \Lambda_0.
$$
Take $\beta$ such that $\beta>\dfrac{\theta+2}{2}$, take $\gamma>0$
such that
$$
\beta\cdot[2\gamma-N(N-1)c2^{\theta-2}k_\infty^\theta]>\gamma(\theta+2).
$$
 Let $g(u)=\gamma \parallel u \parallel^{\theta+2}$, we shall
show that $\{f\leq g\}$ is a set of finite category. We take
$0<\varepsilon<\dfrac{h}{2}$ and $ M>0$ such that $\forall s\in
\mathbb{R}, \gamma|s|^{\theta+2}\leq \varepsilon s^2 +M $.

Define
$$
f_\varepsilon(u)=f(u)-\varepsilon \parallel u \parallel^2
$$

Then
$$
\{f\leq g\} \subset \{ f_\varepsilon(u)\leq M\}
$$
We can use the similar proof of Lemma 3.2 to show that there exists
$k_1\in \mathbb{R} $ such that
$$
\parallel u \parallel\leq k_1p(u), \forall u\in\{f_\varepsilon\leq M\}
$$
 then by Lemma 3.3,
 $$
Cat_{\Lambda_0}\{f\leq g\} \leq Cat_{\Lambda_0}\{f_\varepsilon\leq
M\}<+\infty
 $$

 From Lemma 3.4, we know that $(v)$ is satisfied.

Since $\parallel u \parallel \leq k_1 p(u)$, we take
$\lambda_0=:\gamma(k_1r_2)^{\theta+2}$. Then if $u\in \{f=g
\geq\lambda_0\}$, so
$$
\parallel u \parallel \geq
[\dfrac{\lambda_0}{\gamma}]^{\frac{1}{\theta+2}}=k_1r_2
$$
So we can obtain :
$$
r_2\leq p(u)\leq |u_i(t)-u_j(t)|\leq 2|u|_\infty\leq
2k_\infty\parallel u \parallel
$$

\begin{align}
&\int_0^1 \langle \nabla_u V(u),u\rangle dt  \notag \\
&=\int_0^1\underset{1\leq i<j\leq N}{\sum}\langle\nabla
V_{ij}(u_i-u_j), (u_i-u_j)\rangle dt\notag   \\
&\leq \int_0^1\sum_{1\leq i< j\leq N}c|u_i-u_j|^\theta dt \notag\\ &
\leq c \dfrac{N(N-1)}{2}(2k_\infty\parallel u
\parallel)^\theta \tag{3.7}
\end{align}

Since
\begin{align}
\parallel f'(u) \parallel\parallel u \parallel &\geq \langle f'(u),
u\rangle
\notag\\
&=2f(u)-\dfrac{1}{2}\parallel u  \parallel^2 \int_0^1 \langle
\nabla_u(u),u \rangle dt \tag{3.8}
\end{align}
Then by $(3.7)$ and $(3.8)$, we have
$$\parallel f'(u) \parallel\geq
[\,2\gamma-N(N-1)c2^{\theta-2}k_\infty^\theta\,]\parallel u
\parallel^{\theta+1}  \eqno{(3.9)}
$$
Since $\parallel g'(u) \parallel=\gamma (\theta+2)\parallel u
\parallel^{\theta+1}$,  from our choice of $\gamma$, we have
$$
\beta\parallel f'(u) \parallel-\parallel g'(u) \parallel\geq
0,\forall u\in \{f=g\geq \lambda_0\}   \eqno{(3.10)}
$$
That is,  $(vi)$  holds.
\end{proof}

\end{document}